\documentclass{article}

\usepackage{latexsym}
\usepackage{amsthm}
\usepackage{amsmath}
\usepackage{amssymb}
\usepackage{graphicx}

\newtheorem{proposition}{Proposition}
\newtheorem{theorem}{Theorem}

\theoremstyle{definition}
\newtheorem{definition}{Definition}

\newtheorem{example}{Example}
\newtheorem{remark}{Remark}

\newcommand{\bra}[1]{\langle #1|}
\newcommand{\ket}[1]{| #1 \rangle }
\newcommand{\ip}[2]{{\langle #1|}{ #2 \rangle }}
\newcommand{\tr}[1]{{\rm tr}[#1]}

\newcommand{\be}{\begin{eqnarray}}
\newcommand{\ee}{\end{eqnarray}}

\newcommand{\cE}{{\cal E}}
\newcommand{\cG}{{\cal G}}

\newcommand{\cI}{{\cal I}}

\newcommand{\cN}{{\cal N}}
\newcommand{\cF}{{\cal F}}

\newcommand{\cS}{{\cal S}}
\newcommand{\cH}{{\cal H}}

\newcommand{\cB}{{\cal B}}
\newcommand{\cL}{{\cal L}}

\begin{document}

\title{Open system dynamics of simple collision models
\footnote{Lecture presented at 46 Karpacz Winter School of Theoretical Physics
Ladek Zdroj, Poland, 8 - 13 February 2010}}

\author{M\'ario Ziman$^{1,2}$ and Vladim\'\i r Bu\v zek$^{1,2}$
\\ \\
\it\normalsize
$^1$Institute of Physics, Slovak Academy of Sciences,
D\'ubravsk\'a cesta 9, 
\\ \it\normalsize 845 11 Bratislava, Slovakia,\\ 
\it\normalsize $^2$Faculty of Informatics, Masaryk University, Botanick\'a 68a
\\\it\normalsize  602 00 Brno, Czech Republic
}
\date{}
\maketitle

\begin{abstract}
A simple collision model is employed to introduce elementary
concepts of open system dynamics of quantum systems. In particular,
within the framework of collision models we introduce the 
quantum analogue of thermalization process called 
quantum homogenization and simulate quantum decoherence processes.
These dynamics are driven by partial swaps and controlled unitary
collisions, respectively. We show that collision models can be used
to prepare multipartite entangled states. Partial swap dynamics 
generates $W$-type of entanglement saturating the CKW inequalities,
whereas the decoherence collision models creates $GHZ$-type of entangled
states. The considered evolution of a system in a sequence of 
collisions is described by a discrete semigroup $\cE_1,\dots,\cE_n$.
Interpolating this discrete points within the set of quantum channels 
we derive for both processes the corresponding Linblad master equations.
In particular, we argue that collision models can be used as simulators
of arbitrary Markovian dynamics, however, the inverse is not true.
\end{abstract}



\section{Open system dynamics}\label{aba:sec1}
The goal of these lectures is to present the elementary ideas and features
of dynamics of open quantum systems by analyzing the properties of the 
simplest collision model we can think of. The material is based on papers
\cite{ziman_homogenization,scarani_thermo,ziman_open,ziman_decoherence}
on collision models coauthored by authors. It is not considered 
as a review paper on collision models or open system dynamics. 

It turns out it is surprisingly useful in physics to distinguish 
the concepts of isolated and open systems. 
The concept of isolated systems represents a simplification of physical 
reality, which serves as a playground for a clear formulation of 
elementary physical principles. It is postulated that the 
dynamics of isolated quantum systems is driven by \emph{Schr\"odinger 
equation} \cite{schrodinger1926}
\be
i\hbar\frac{d}{dt}\psi_t=H_t\psi_t\, ,
\ee
where $H_t$ is a hermitian operator called \emph{Hamiltonian}.
As a consequence it follows that the state transformations are 
unitary, i.e. $\psi_t\to\psi_{t^\prime}=U_{t\to t^\prime}\psi_t$, where
$U_{t\to t^\prime}U_{t\to t^\prime}^\dagger=U_{t\to t^\prime}^\dagger U_{t\to t^\prime}=I$.
If $H_t$ is time-independent, i.e. $H_t\equiv H$, then 
$U_{t\to t^\prime}=\exp{[-\frac{i}{\hbar}H(t^\prime-t)]}$. We can
define $U_t=\exp{(-\frac{i}{\hbar}Ht)}$ and write 
$U_{t\to t^\prime}=U_{t^\prime} U_t^\dagger$, $\psi_t=U_t\psi_0$.

Interactions between a system under consideration and its environment
are responsible for the violation of system's (dynamical) isolation. The 
joint evolution of the system and the environment will be still unitary, 
however, the system itself undergoes a different dynamics. The point is 
that these interactions, because of the complexity of the environment, 
are typically out of our control. And our wish is to invent models of 
system-environemt interactions capturing faithfully the key properties 
of the dynamics of the system alone. For a general discussion on 
the models of open system dynamics we refer to \cite{davies1970,alicki,breuer}.
In what follows we will restrict our analysis to a 
particular toy model of the open system dynamics. 

In the rest of this section we will introduce the model and elementary
properties of quantum channels. In the following two sections we will
discuss thermalization and decoherence processes within this model. In Section 
IV we will focus on entanglement created in the discussed collision models and
finally, in the Section V we will derive master equations for these collision
processes.

\subsection{Simple collision model}
Let $\cH$ be the Hilbert space of the system of interest. States are identified
with elements of the set of density operators (positive operators of unit trace)
\be
\cS(\cH)=\{\varrho: \varrho\geq O,\ \tr{\varrho}=1\}\,.
\ee
We will assume that environment consists of huge number of particles, each 
of them described by the same Hilbert space $\cH$. In a sense, 
it forms a reservoir of particles. Moreover, we assume that initially they are
all in the same state $\xi$, i.e. the initial state of the 
environment/reservoir is $\omega=\xi^{\otimes N}$. In each time step the 
system interacts with
a single environment's particle. The interaction is described by a unitary
operator $U$ acting on joint Hilbert space $\cH\otimes\cH$. Thus, the
one-step system's evolution is described by a map
\be
\label{eq:ucpmodel}
\varrho\to\varrho^\prime=
\cE[\varrho]={\rm tr}_{\rm env}[U(\varrho\otimes\xi)U^\dagger]\,,
\ee
where $\varrho$ is the initial state of the system and ${\rm tr}_{\rm env}$
denotes the partial trace over the environment. In each time step a system
undergoes a collision with a single particle from the reservoir. In a realistic
scenarios these collisions are happening randomly, but we will assume that
each reservoir's particle interacts with the system at most once 
(see Fig. \ref{fig1}). Thus, the system evolution is driven by a 
sequence of independent collisions. After the $n$th collision we get
\be
\varrho_0\to\varrho_n^\prime=
{\rm tr}_{\rm env}[U(\varrho_{n-1}\otimes\xi)U^\dagger]=
\cE[\varrho_{n-1}]=\cE^n[\varrho_0]\,.
\ee

\begin{figure}
\begin{center}
\includegraphics[width=8cm]{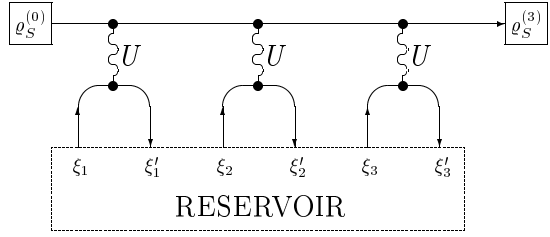}
\end{center}
\caption{A simple collision model.}
\label{fig1}
\end{figure}

\subsection{Quantum channels}

The mapping $\cE$ defined on the set of all operators 
$\cL(\cH)$ satisfies the following properties:
\begin{enumerate}
\item $\cE$ is linear, i.e. $\cE[X+cY]=\cE[X]+c\cE[Y]$ for all 
$X,Y\in\cL(\cH)$ and all complex numbers $c$.
\item $\cE$ is trace-preserving, i.e. $\tr{\cE[X]}=\tr{X}$ for all 
operators $X$ with finite trace.
\item $\cE$ is completely positive, i.e. 
$(\cE\otimes\cI)[\Omega]\geq O$ for all positive operators 
$\Omega\in\cL(\cH\otimes\cH)$.
\end{enumerate}
Due to \emph{Stinespring dilation theorem} the inverse is also true. 
If a mapping $\cE$ satisfies above three conditions, 
there is a unitary operator $U$ and a state $\xi$ such that the 
Eq.~\eqref{eq:ucpmodel} holds. Such completely positive trace-preserving 
linear maps we call \emph{channels}.

For any state $\xi\in\cS(\cH)$ there exist a unit vector $\Psi\in\cH\otimes\cH$
such that $\xi={\rm tr}_2 \ket{\Psi}\bra{\Psi}$, where ${\rm tr}_2$ is the 
partial trace over the second subsystem. We call $\Psi$ 
\emph{a purification} of $\xi$. It is straightforward to verify that
the following identity holds
\be
\cE[\varrho]={\rm tr}_{\rm env}[U(\varrho\otimes\xi)U^\dagger]=
{\rm tr}_{\rm env^\prime}[(U\otimes I)
(\varrho\otimes\ket{\Psi}\bra{\Psi})(U\otimes I)^\dagger]\,.
\ee
Let $\varphi_j$ is an orthonormal basis on $\cH\otimes\cH$ and 
define linear operators $A_j=\ip{\varphi_j}{(U\otimes I)\Psi}$ 
acting on the system. Then the right-hand side of the above 
equation reads
\be
\cE[\varrho]=\sum_j \ip{\varphi_j}{(U\otimes I)\Psi}\varrho
\ip{(U\otimes I)\Psi}{\varphi_j}=\sum_j A_j\varrho A_j^\dagger\,.
\ee
By definition the so-called \emph{Kraus operators} 
$A_j$ satisfy the normalization $\sum_j A_j^\dagger A_j=I$. 
Any mapping of the Kraus form $\cE[\varrho]=\sum_j A_j\varrho A_j^\dagger$
for arbitrary set of operators $\{A_1,A_2,\dots\}$ fulfilling the
normalization constraint determines a valid quantum channel.

Channels with a single Kraus operator $\cE[\varrho]=A\varrho A^\dagger$ 
are \emph{unitary channels}, because the normalization
$A^\dagger A=I$ implies that $A$ is unitary. Let us note that Kraus
representation and also the number of Kraus operators are not unique.
For more details we refer to \cite{nielsen,heinosaari}.
\section{Quantum homogenization as an analogue to thermalization}
\label{aba:sec2}
\emph{The 0th law of thermodynamics} postulates that an object brought 
in a contact with a reservoir of temperature $T$ will equalize its 
temperature with the reservoir's temperature. This process is known 
as the \emph{process of thermalization}. Our attempt is to design 
a quantum analogue of this process.

The concept of  temperature is the key ingredient of the thermalization
process. In fact, as a consequence it is the temperature  that describes 
a state of the system, that is, it contains the relevant part
of the information needed for the investigation of the system's properties. 
What is the quantum analogue of the temperature? There are several ways 
how this concept can be introduced. One can use the connection between 
temperature and the entropy, or energy per single particle. By choosing 
such analogue, the thermalization can be defined as a process
that leads a system with a given temperature $T_S$ to the temperature
of the reservoir $T_\xi$. In classical theory this process
is achieved because of mutual interactions/collisions between the particles.
In the process of collision they mutually (partly) exchange their 
microscopic properties (energy, momentum). Subsequently, the temperature
of the system being in a contact with the reservoir is changing and finally 
(approximately) equalize  $T_\xi$. Let us note that we have used 
the temperature as a property that can be attributed even to a single particle,
however, such definition is meaningful only in the statistical sense. 

In quantum theory the state of a quantum system is described with
a density operator. This operator contains all the information. 
Therefore, we choose single-particle density operator
to play the role of the quantum temperature in our quantum analogue.
So the goal is to design a model of the dynamics
of the system+reservoir such that the initial state
of the system particle $\varrho_S$ evolves into the
state of the reservoir particles $\xi$. Because of the property
that finally all particles are described by the same state $\xi$
we will refer to this process as to {\it homogenization} instead
of the thermalization. In fact this notion is more appropriate
to characterize the features of our model.

Ideally the homogenization should implement the following transformation
\be
\varrho_S\otimes\xi\otimes\cdots\otimes\xi
\mapsto\xi_S\otimes\xi\otimes\cdots\otimes\xi\,,
\ee
for all $\varrho$ and $\xi$. For finite reservoirs this process will perform
$N\to N+1$ cloning transformation, which would violate the quantum no-cloning 
theorem \cite{buzek_cloning,scarani_cloning}. Therefore, this process cannot
be achieved in any dynamical model respecting the quantum theory. We will 
relax our requirements and investigate whether the homogenization 
can be implemented approximately via a sequence of (independent) collisions.
Moreover, if initially $\varrho=\xi$, meaning there is nothing to homogenize,
then the evolution is trivial. Let us summarize our assumptions and problem.

\begin{definition}
We call an interaction $U$ a generator of $\delta$-homogenization $\delta>0$
is the following conditions are satisfied
\begin{itemize}
\item{\it Trivial homogenization.} 
The joint state of the system and the environment after $n$th interaction
reads
\be
\Omega_n=U_n\cdots U_1(\varrho_S\otimes\xi^{\otimes n}) 
U_1^\dagger\cdots U_n^\dagger\,,
\ee
where $U_j=U_{S,j}\otimes I_{{\rm env}\setminus\{j\}}$ describes the $j$th 
interaction
($I_{{\rm env}\setminus\{j\}}$ is the identity operator on the environment 
excluding the $j$th particle). The for all $\xi\in\cS(\cH)$
\be
\xi_S\otimes\xi\otimes\cdots\otimes\xi
\mapsto\Omega_n=\xi_S\otimes\xi\otimes\cdots\otimes\xi\,.
\ee
\item{\it Convergence.} Let $\varrho_S^{(n)}$ be the state of the system
after $n$th collision. Then there exist $N_\delta$ such that for all $n>N_\delta$
\be
D(\varrho_S^{(n)},\xi)\leq\delta\,.
\ee
\item{\it Stability of the reservoir.}
Let $\xi_j^\prime$ be the state of the $j$th reservoir's particle after the 
its interaction with the system. Then 
\be
D(\xi_j^\prime,\xi)\leq\delta\,,
\ee
for all $j$.
\end{itemize}
\end{definition}

Our task is simple: find $U$ satisfying the above properties.

\subsection{Trivial homogenization}
Let us denote by $\cH_\pm$ the symmetric and antisymmetric subspaces
of $\cH\otimes\cH$. A vector $\Phi$ belongs to $\cH_\pm$ if $S\Phi=\pm\Phi$, 
where $S$ is the swap operator defined as $S=\sum_{jk} 
\ket{\varphi_j\otimes\varphi_k}\bra{\varphi_k\otimes\varphi_j}$
for arbitrary orthonormal basis $\{\varphi_j\}$ of $\cH$. It follows
that for all $\varrho,\xi\in\cS(\cH)$ 
\be
S(\varrho\otimes\xi) S^\dagger=\xi\otimes\varrho\,.
\ee
Let us denote by $P_\pm$ projections onto symmetric and antisymmetric 
subspace. Then $S=P_+-P_-$.

The condition of trivial homogenization implies that
\be
U\ket{\psi\otimes\psi}=\ket{\psi\otimes\psi}
\ee
for all unit vectors $\psi\in\cH$. This determines the action 
of the transformation $U$ on the symmetric subspace $\cH_{+}$.
However, it does not give any constraint on the antisymmetric 
subspace $\cH_{-}$. Since $\cH\otimes\cH=\cH_+\oplus\cH_-$ we can write
$U=e^{i\gamma}I_+\oplus V_-$, where $V_-$ is arbitrary unitary 
operator on the subspace
$\cH_-$ and $I_+$ is the identity operator on the symmetric subspace $\cH_+$.
In fact, $I_+=P_+|_{\cH_+}$. Any of the collisions of the form 
$U=e^{i\gamma}I_+\oplus V_-$ fulfills the conditions of trivial 
homogenization.

Let us consider the case of $d={\rm dim}\cH=2$ (qubit), for which
the antisymmetric space is one dimensional. In one-dimensional Hilbert
spaces the unitary operators are complex square roots of the unity, i.e.
$V_-=e^{i\beta}I_-$. Explicitly the transformation can be expressed as
\begin{eqnarray*}
U \ket{00} & = & e^{i\gamma}\ket{00}\; , \\
U \ket{11} & = & e^{i\gamma}\ket{11}\; , \\
U( \ket{01}+\ket{10}) & = & e^{i\gamma}( \ket{01}+\ket{10})  \; , \\
U(\ket{01}-\ket{10}) & = & 
e^{i(\gamma+\beta)}(\ket{01}-\ket{10})  \; . 
\end{eqnarray*}
Using the identities $e^{i\gamma}=e^{i(\gamma+\beta/2)}e^{-i\beta/2}$,
$e^{i(\gamma+\beta)}=e^{i(\gamma+\beta/2)}e^{i\beta/2}$ and discarding
the irrelevant global phase factor $e^{i(\gamma+\beta/2)}$ we get
a one-parametric set of solutions
\be
U_\eta=\cos\eta I +i\sin\eta S\,,
\ee
where we set $\eta=-\beta/2$.

In case of more dimensional space (where $d>2$) the set of unitary 
transformations is larger. However, when we fix $V_+=e^{i\beta}I_-$,
we get the same one-parametric set of solutions
\be
U_\eta=\cos\eta I +i\sin\eta S=e^{i\eta S}\,.
\ee
Any element of this class of unitary operators we will call 
\emph{a partial swap}. In what follows we will focus on the collision
dynamics generated by partial swap collisions.

\subsection{Partial swap collisions}
In what follows we denote $\sin \eta=s$ and $\cos  \eta=c$. In the process 
of homogenization, the system particle interacts
sequentially with one of the $N$ particles of the reservoir through
the transformation $U_\eta=cI+is S$. After the first collision
the system and the first reservoir particle undergo the evolution
\be
\nonumber
\varrho_S^{(0)}\otimes\xi\mapsto U_\eta(\varrho_S^{(0)}\otimes\xi)U_\eta^\dagger=
c^2\varrho_S^{(0)}\otimes\xi+s^2\xi\otimes\varrho_S^{(0)}+ics[S,\varrho_S^{(0)}\otimes\xi]\,.
\ee
The states of the system particle and of the reservoir
particle are obtained as partial traces. Specifically, after
the system particle is in the state described by the density operator
\be
\label{2.4}
\varrho_S^{(1)}=c^2\varrho_S^{(0)}+s^2\xi+ics[\xi,\varrho_S^{(0)}],
\ee
while the first reservoir particle is now in the state
\be
\label{2.5}
\xi_{1}^\prime=s^2\varrho_S^{(0)}+c^2\xi+ics[\varrho_S^{(0)},\xi].
\ee
When we recursively apply  the partial-swap transformation then 
after the interaction with
the $n$-th reservoir particle, we obtain
\be
\label{2.6}
\varrho_S^{(n)}=c^2\varrho_S^{(n-1)}+s^2\xi+ics[\xi,\varrho_S^{(n-1)}]\,
,
\ee
as the expression for the density operator of the system
particle, while the $n$-th reservoir particle is in the state
\be
\label{2.7}
\xi_{n}^\prime=
s^2\varrho_S^{(n-1)}+c^2\xi+ics[\varrho_S^{(n-1)},\xi]\, .
\ee

\subsection{System's convergence}
Our next task is to analyze the remaining conditions on the homogenization
process. First task is to show that $\varrho^{(n)}_S$ monotonically converges 
to a $\delta$-vicinity of $\xi$ for all parameters $\eta\neq 0$. We will 
utilize the {\it Banach fixed point theorem} \cite{Simon&Reed80} 
claiming that iterations
of a \emph{strictly contractive} mapping converge to the single-point 
contraction. This point is unique and determined by the condition
$\cE[\xi]=\xi$. A channel ${\cal E}$ is called {\em strictly contractive} 
if it fulfills  the inequality $D({\cal E}[\varrho],{\cal E}[\xi])
\le kD(\varrho,\xi)$ with $0\le k<1$ for all $\varrho,\xi\in\cal S$.

First of all, it is easy to verify that $\xi$ is a fixed point of 
$\cE_\xi$, where $\cE_\xi$ is defined by the Eq.(\ref{2.4}).
As a distance we will use the Hilbert-Schmidt distance
\be
D(\varrho_1,\varrho_2)=||\varrho_1-\varrho_2||=
\sqrt{\tr{\varrho_1^2}+\tr{\varrho_2^2}
-2\tr{\varrho_1\varrho_2}}
\ee
where $\varrho_1,\varrho_2$ represent quantum states. Consequently,
\be
D^2(\cE[\varrho_1],\cE[\varrho_2])
=c^4 D^2(\varrho_1,\varrho_2)+c^2s^2||[\xi,\varrho_1-\varrho_2]||^2
\ee
where we used that $\cE[\varrho]=c^2\varrho+s^2\xi+ics[\xi,\varrho]$.
Defining an operator $\Delta=\varrho_1-\varrho_2$
our aim is to prove the relation $||[\xi,\Delta]||^2\le ||\Delta||^2$.
Assuming the spectral decomposition of $\xi$ equals 
$\xi=\sum_k\lambda_k\ket{k}\bra{k}$ and performing the trace in the basis 
$\ket{k}$ we obtain
  \be
  \nonumber
  ||\xi\Delta-\Delta\xi||^2&=&\tr{(\xi\Delta-\Delta\xi)^\dagger
  (\xi\Delta-\Delta\xi)}=2\tr{\xi^2\Delta^2}-2\tr{(\xi \Delta)^2}\\
  \nonumber &=&2{\mathbf \sum}_j \lambda_j^2 \bra{j}\Delta^2\ket{j}
  -2{\mathbf \sum}_{j,k}\lambda_j\lambda_k |\bra{j}\Delta\ket{k}|^2\\
  \nonumber &=& 2{\sum}_{j,{k}} \lambda_j^2 \bra{j}\Delta{\ket{k}\bra{k}}\Delta\ket{j}-2{\mathbf \sum}_{j,k}\lambda_j\lambda_k |\bra{j}\Delta\ket{k}|^2\\
  \nonumber &=& {\mathbf \sum}_{j,k} ({2\lambda_j^2}-2\lambda_j\lambda_k)|\bra{j}\Delta\ket{k}|^2\\
  \nonumber &=& {\mathbf \sum}_{j,k}({\lambda_j^2+\lambda_k^2}-2\lambda_j\Lambda_k)|\bra{j}\Delta\ket{k}|^2={\mathbf \sum}_{j,k}(\lambda_j-\lambda_k)^2|\bra{j}\Delta\ket{k}|^2
  \\ \nonumber &\leq& {\mathbf \sum}_{j,k}|\bra{j}\Delta\ket{k}|^2=
  \tr{\Delta^\dagger\Delta}=||\Delta||^2
  \ee
where we employed the hermiticity of $\Delta$ 
($|\bra{j}\Delta\ket{k}|^2=|\bra{l}\Delta\ket{j}|^2$) and 
positivity of $\xi$ which means $|\lambda_1-\lambda_2|\le 1$.
The proved inequality enables us to write
\be
D(\cE[\varrho_1],\cE[\varrho_2])\le |c| D(\varrho_1,\varrho_2)
\ee
i.e. the contractivity coefficient 
$k$ is determined by the parameter
$c=\cos\eta$ of the partial swap 
and the map $\cE$ is contractive whenever $|c|<1$.
This result is important because it ensures that
in the limit of infinite steps the system will 
be not only $\delta$-homogenized, but will be described exactly 
by the state $\xi$ of the reservoir. It means the distance 
$D(\varrho_S^{(n)},\xi)\to 0$ vanishes with the number of 
interactions $n$.

\subsection{Stability of the reservoir}
In order to satisfy the last condition (stability of the reservoir)
we need to evaluate the distances $D(\xi_j^\prime,\xi)$. For each
value of $\eta$ we can specify $\delta$  and $N_\delta$. However, the
question is whether for arbitrary value of $\delta>0$ there is a 
suitable collision $U_\eta$ and how large is the value 
of $N_\delta$ for which $D(\varrho_S^{(N_\delta)},\xi)\leq \delta$.

For the system we can combine the contractivity bound 
$D(\cE[\varrho],\cE[\xi])\leq |c|D(\varrho,\xi)$ with the fact that
$\cE[\xi]=\xi$ to obtain an estimate
\be
D(\xi,\varrho_S^{(n)})\leq |c|^n D(\xi,\varrho_S^{(0)})=\sqrt{2}|c|^n\leq\delta
\,.
\ee
Consequently,
\be
N_\delta=\frac{\ln(\delta/\sqrt{2})}{\ln|c|}\,,
\ee
but the potential values of $\delta$ we need to specify from
the stability of the reservoir. A direct calculation gives the following bound
\be
\nonumber
D(\xi,\xi_n^\prime)
&=&||\xi-c^2\xi-s^2\varrho_S^{(n-1)}-ics[\varrho_S^{(n-1)},\xi]||\\
\nonumber
&\leq& s^2||\xi-\varrho_S^{(n-1)}||+|cs|\cdot||[\varrho_S^{(n-1)},\xi]||\\
\nonumber
&\leq& s^2D(\xi,\varrho_S^{(n-1)})+2|cs|\cdot||\varrho_S^{(n-1)}||\cdot||\xi||\\
\nonumber
&\leq& \sqrt{2}|sc| (|s|\cdot|c|^{n-2}+\sqrt{2})\\
&<& \sqrt{2}|sc| (1+\sqrt{2})\equiv\delta\,.
\ee
Important is that $\delta$ can be arbitrarily small, hence 
$\delta$-homogenization is achievable for a restricted class 
of collisions $U_\eta$ satisfying the identity 
$(2+\sqrt{2})|\sin\eta\cos\eta|\leq\delta$. The number of steps
$N_\delta$ we achieve if the value of $\delta$ is inserted into the
formula for $N_\delta$, i.e.
\be
N_\delta\approx\frac{\ln{[(1+\sqrt{2})|sc|]}}{\ln|c|}\, .
\ee
Let us note that one can derive more precise expression for
$N_\delta$ and $\delta$, but this is not important for our further
purposes.

In summary, the class of partial swap operators $U_\eta$ satisfies
the conditions for $\delta$-homogenization process. 

\begin{remark}{\it Swap in the classical theory.}
Let us note that in the classical picture of thermalization process
we also use model of mutual collisions of the system with its 
thermal reservoir. In these (inelastic) collisions the particles 
exchange energy and momentum. In some sense, this process 
can be understood as a transformation that partially ``swaps'' energy 
and momentum of involved particles. Thus, from this perspective the 
homogenizing properties of partial swap collisions are not that 
surprising and can be viewed as the quantum analogue of classical
collisions.
\end{remark}

\subsection{Invariance of single-particle average state}
In this section we will allow particles forming the reservoir interacts
among each other, but interactions are always pairwise. 
Partial swap operators possess one unique property, which
is important from the point of view of more realistic thermalization 
process and reservoir's stability. Consider $\Omega$ being an $n$-partite 
composite system. It can be written in the form
\be
\Omega=\varrho_1\otimes\cdots\otimes\varrho_n+\Gamma\,,
\ee
where $\Gamma$ is a traceless hermitian operators and
$\varrho_j={\rm tr}_{\overline{j}}\Omega$ are the states of
individual subsystems. Applying a partial swap collision
between $j$th and $k$th subsystems we obtain a state
\be
\Omega^\prime&=&c^2\Omega+s^2 S_{jk}\Omega S_{jk}+ics[\Omega,S_{jk}]\\
&=& \varrho_1^\prime\otimes\cdots\otimes\varrho^\prime_n+\Gamma^\prime\,,
\ee
where $S_{jk}$ denotes the swap operator between the subsystems $j$ and $k$,
\be
\varrho_j^\prime&=&c^2\varrho_j+s^2\varrho_k+ics[\varrho_j,\varrho_k]\\
\varrho_k^\prime&=&c^2\varrho_k+s^2\varrho_j-ics[\varrho_j,\varrho_k]\\
\varrho_{l}^\prime&=&\varrho_l\quad{\rm for}\quad l\neq j,k\,,
\ee
and $\Gamma^\prime$ is the traceless part. It is straightforward to
see that
\be
\overline{\varrho}_{\rm one}=\frac{1}{n}\sum_j \varrho_j=\frac{1}{n}
\sum_j\varrho_j^\prime=\overline{\varrho}^\prime_{\rm one}\,,
\ee
thus, in the collision model provided by partial swap collisions
the \emph{average one-particle state} $\overline{\varrho}_{\rm one}$ 
is preserved. If taken the single-particle state as the quantum analogue 
of temperature, then this property means that partial swap interactions
preserves the "quantum'' temperature.

\section{Quantum decoherence via collisions}\label{aba:sec3}
In this section we will focus on quantum dynamics of decoherence. Let us note
that in literature on quantum information the decoherence is used to name
any nonunitary channel and the channels we are interested in are called
(generalized) phase-damping channels. In accordance with our
previous work \cite{ziman_decoherence} we will understand by decoherence 
a channel with the following properties:
\begin{itemize}
\item Preservation of the diagonal elements of a density matrix with respect
to a given (decoherence) basis. Let us denote by 
$\cB=\{\varphi_1,\dots,\varphi_d\}$ the \emph{decoherence basis} of $\cH$.
Then
\be
\cE\circ{\rm diag}_\cB={\rm diag}_\cB\,,
\ee
where ${\rm diag}_\cB$ is a channel diagonalizing density operators
in the basis $\cB$, i.e. $\varrho\to{\rm diag}_\cB[\varrho]$.
\item Vanishing of the off-diagonal elements of a density matrix with respect
to a given (decoherence) basis. In particular, asymptotically the iterations
of the decoherence channels results in the diagonalization of the density 
matrix, i.e.
\be
\lim_{n\to\infty}\cE^n={\rm diag}_\cB\,.
\ee
\end{itemize}
Our tasks is similar as in the previous part: find collisions $U$ generating
the decoherence in basis $\cB$.

The basis preservation implies that
\be
\ket{\varphi_j\otimes\varphi_k}\mapsto 
e^{i\eta_{jk}}\ket{\varphi_j\otimes\varphi_{jk}^\prime}\,,
\ee
where $\eta_{jk}$ are phases and for each $j$ the unit vectors
$\{\varphi_{jk}^\prime\}_k$ form an orthonormal basis of $\cH$.
It is not difficult to verify that the unitary operators describing
the collisions have the following form 
\be
U=\sum_j \ket{\varphi_j}\bra{\varphi_j}\otimes V_j\,,
\ee
where $V_j=\sum_k e^{i\eta_{jk}}\ket{\varphi_{jk}^\prime}\bra{\varphi_k}$
are unitary operators on individual systems of the environment.
Such operators form a class of \emph{controlled unitaries} with the system
playing the role of the control system and the environment playing the role of
the target system.

A single controlled unitary collision induce the following channel
on the system
\be
\cE[\varrho]={\rm tr}_{\rm env}[U(\varrho\otimes\xi) U^\dagger]
=\sum_{j,k}\ket{\varphi_j}\bra{\varphi_j}\varrho\ket{\varphi_k}\bra{\varphi_k}
\tr{V_j\xi V_k^\dagger}\,.
\ee
Let us note an interesting feature that under the action of this channel
the matrix elements of density operators do not mix, i.e. the output 
value of $\varrho_{jk}^\prime=\ip{\varphi_j}{\cE[\varrho]\varphi_k}$ 
depends only on the value of $\varrho_{jk}=\ip{\varphi_j}{\varrho\varphi_k}$.
In order to fulfill the second decoherence condition it is sufficient to
show that for $j\neq k$
\be
|\varrho_{jk}^\prime|<|\varrho_{jk}|\,.
\ee

For any unitary operator $W=\sum_l e^{iw_l}\ket{l}\bra{l}$ 
\be
|\tr{\xi W}|=|\sum_l e^{iw_l}\bra{l}\xi\ket{l}|=|\sum_l p_l e^{iw_l}|\leq 1\,,
\ee
where $p_l=\bra{l}\xi\ket{l}$ are probabilities. This inequality is saturated
only if $\xi$ is an eigenvector of $W$. It follows that the inequality 
\be
|\varrho_{jk}^\prime|
=|\varrho_{jk}|\cdot|\tr{V_j\xi V_k^\dagger}|\leq |\varrho_{jk}|\,
\ee
holds, hence the off-diagonal elements are non-increasing.
They are strictly decreasing if and only if $|\tr{V_j\xi V_k^\dagger}|<1$
for all $j\neq k$, which is the case if $\xi$ is not an eigenvector
of any of the operators $V_k^\dagger V_j$.

In summary, we have shown that decoherence processes are intimately
related with controlled unitary operators. In particular,
the decoherence processes are generated by controlled unitary
collisions $U=\sum_j \ket{\varphi_j}\bra{\varphi_j}\otimes V_j$
satisfying $|\tr{V_j\xi V_k^\dagger}|<1$ for all $j\neq k$. In such case
\be
\lim_{n\to\infty} \cE^n[\varrho]={\rm diag}_\cB[\varrho]\,,
\ee
where ${\rm diag}_\cB$ is a channel diagonalizing the input density
operator in the basis $\cB$. Let us note that there are always
states $\xi$ (eigenvectors of $V_k^\dagger V_j$) for which the decoherence 
is not achieved, because for them some of the off-diagonal
elements are preserved. On the other side, for each controlled unitary
collision $U$ there are always states of the environment $\xi$ 
inducing the decoherence of all off-diagonal matrix elements.

\subsection{Simultaneous decoherence of the system and the environment}
Could it happen that both the environment and the system are decohering
simultaneously? In the same decoherence basis? We will see that answers
to both these questions are positive.

For the considered
collision model driven by controlled unitary interactions we obtain that
the evolution of the particles in the reservoir is described
by the following channels
\be
\cN[\xi]={\rm tr}_{\rm sys}[U(\varrho\otimes\xi)U^\dagger]=
\sum_j \varrho_{jj} V_j\xi V_j^\dagger\,.
\ee
Such channels are known as \emph{random unitary channels}.

It is known that random unitary channels can describe decoherence
channels and it was a surprising result that decoherences are not 
necessarily random unitary channels \cite{dariano_decoherence}.
We do not need to go into the details of this relation to see that
both the system and the environment can decohere. It follows from
our previous discussion that environment will decohere if
$U$ is controlled unitary transformation of the form
\be
U=\sum_j U_j\otimes\ket{\psi_j}\bra{\psi_j}\,,
\ee
where $\{\psi_j\}$ is the decoherence basis of the particles in the 
environment. That is, the question is whether there are unitary operators
that can be written as 
\be
U=\sum_j U_j\otimes\ket{\psi_j}\bra{\psi_j}=
\sum_j \ket{\varphi_j}\bra{\varphi_j}\otimes V_j\,,
\ee
for suitable decoherence bases $\{\varphi_j\}$, $\{\psi_j\}$ 
and unitary operators $U_j,V_j$.

\begin{example}{\it CTRL-NOT.}
\be
U_{\rm ctrl-NOT}&=&\ket{0}\bra{0}\otimes I+\ket{1}\bra{1}\otimes\sigma_x=
I\otimes\ket{+}\bra{+}+\sigma_z\otimes\ket{-}\bra{-}\\
&=& \ket{0+}\bra{0+}+\ket{0-}\bra{0-}+\ket{1+}\bra{1+}-\ket{1-}\bra{1-}
\,,
\ee
where $\ket{\pm}=(\ket{0}\pm\ket{1})/\sqrt{2}$.
\end{example}

\begin{proposition}
A unitary operator $U$ describes a collision with simultaneous
decoherence of both interacting systems if and only if
$U=\sum_j\ket{\varphi_j}\bra{\varphi_j}\otimes V_j$ and the 
unitary operators $V_j$ commute with each other.
\end{proposition}
\begin{proof}
By definition, the simultaneous decoherence requires
simultaneous preservation of decoherence bases $\{\varphi_j\}$ 
and $\{\psi_k\}$, i.e.
\be
\ket{\varphi_j\otimes\psi_k}\to e^{i\eta_{jk}}\ket{\varphi_j\otimes\psi_k}\,.
\ee
It follows that
\be
U=\sum_{jk}e^{i\eta_{jk}}\ket{\varphi_j}\bra{\varphi_j}\otimes\ket{\psi_k}
\bra{\psi_k}\,,
\ee
and we can define
\be
U_k= \sum_j e^{i\eta_{jk}} \ket{\varphi_j}\bra{\varphi_j}\,,\qquad
V_j= \sum_k e^{i\eta_{jk}} \ket{\psi_k}\bra{\psi_k}\,,
\ee
to get the required expression
\be
U=\sum_k U_k\otimes\ket{\psi_k}\bra{\psi_k}=\sum_j \ket{\varphi_j}
\bra{\varphi_j}\otimes V_j\,.
\ee
It is straightforward to see that operators $\{V_j\}$
and $\{U_k\}$ form sets of mutually commuting elements. 

If the unitary operators $V_j$ commute with each other, then
in their spectral form $V_j=\sum_j e^{i\eta_{jk}}\ket{\psi_k}\bra{\psi_k}$.
Clearly, $U=\sum_j \ket{\varphi_j}\bra{\varphi_j}\otimes V_j=
\sum_k U_k\otimes\ket{\psi_k}\bra{\psi_k}$ with 
$U_k=\sum_j e^{i\eta_{jk}}\ket{\varphi_j}\bra{\varphi_j}$. Therefore, $U$
generates decoherence in both interacting systems.
\end{proof}

Setting $\psi_k\equiv\varphi_k$ the system's decoherence basis and the
environment's decoherence basis coincide. In such case the system and
the environment will decohere with respect to the same basis. An example
of such interaction for the case of qubit is the CTRL-Z operator
\be
U_{\rm ctrl-Z}=\ket{0}\bra{0}\otimes I+\ket{1}\bra{1}\otimes\sigma_z=
I\otimes\ket{0}\bra{0}+\sigma_z\otimes\ket{1}\bra{1}\,,
\ee
generating decoherence in the computational basis $\ket{0},\ket{1}$.
Moreover, this collision is symmetric with respect to exchange of the role
of the control and the target.

\section{Entanglement in collision models}\label{aba:sec4}
The phenomenon of quantum entanglement is considered as one of the key
properties of quantum theory. Our aim is not to discuss the philosophical
background of this concept \cite{horodecki_ent,plenio_ent}, but rather 
focus on the dynamics of entanglement in the considered collision models.
We say a state is \emph{separable} if it can be written as a convex combination
of factorized states, i.e.
\be
\varrho_{\rm sep}=\sum_j p_j \xi_1^{(j)}\otimes\cdots\otimes\xi_n^{(j)}\,.
\ee
A state is \emph{entangled} if it is not separable. The question whether
a given state is entangled, or separable turns out to be very difficult.
Partial results are achieved for the system of qubits and therefore
in this section we will restrict our discussion only to qubits. However,
some of the qualitative properties will be extendible also to larger
systems.

A {\it tangle} is a measure of bipartite entanglement defined as
\be
\tau(\omega)=\min_{\omega=\sum_j p_j\ket{\psi_j}\bra{\psi_j}} 
\sum_j p_j \tau(\psi_j)\,,
\ee
where $\tau(\psi)=S_{\rm lin}({\rm tr}_1{\ket{\psi}\bra{\psi}})=
2(1-\tr{({\rm tr}_1{\ket{\psi}\bra{\psi}})^2})=4\det{\rm tr}_1
[{\ket{\psi}\bra{\psi}}]$ is the 
\emph{linear entropy} of the state of one of the subsystems.
Let us note that for pure bipartite states
$S_{\rm lin}({\rm tr}_1{\ket{\psi}\bra{\psi}})=
S_{\rm lin}({\rm tr}_2{\ket{\psi}\bra{\psi}})$. Due to seminal paper of
Wootters \cite{wootters} $\tau(\omega)=C^2(\omega)$, where $C(\omega)$
is the so-called \emph{concurrence} 
\be
C=\max\{0,2\max\{\sqrt{\lambda_j}\}-\sum_j\sqrt{\lambda_j}\}\,,
\ee 
where $\{\lambda_j\}$ are the eigenvalues of 
$R=\omega(\sigma_y\otimes\sigma_y)\omega^*(\sigma_y\otimes\sigma_y)$.

For multiqubit systems the tangle satisfies the so-called 
\emph{monogamy relation} (originally conjectures by Coffman et al.\cite{coffman}
and proved by Osborne et al.\cite{osborne})
\be
\sum_{k, k\neq j} \tau_{jk}\leq \tau_{j}\,,
\ee
where $\tau_j\equiv \tau_{j\overline j}$ and $\overline j$ denotes 
the set of all qubits except the $j$th one, i.e. $\tau_{j}$ is 
a tangle between the $j$th qubit and rest of the qubits considered 
as a single system. The above inequality is called \emph{CKW inequality}.
If the multiqubit system is in a pure state $\Psi$ then
\be
\tau_{j}=4\det{\rm tr}_{\overline{j}}[\ket{\Psi}\bra{\Psi}]\,.
\ee

\begin{example}{\it Saturation of CKW inequalities.}\label{ex:ckw}
At the end of the original paper \cite{coffman} it is stated that states 
in the subspace covered by the basis 
$\{\ket{1}_j\otimes\ket{0^{\otimes N-1}}_{\overline{j}}\}$ saturates 
the CKW inequalities. In what follows we will show
that any state of the form
\be
\label{W}
\ket{\Psi}=\alpha_0 \ket{0^{\otimes N}}+\sum_{j=1}^{N}\alpha_j\ket{1}_j\otimes \ket{0^{\otimes(N-1)}}
\ee
saturates the CKW inequalities.

The reduced bipartite density matrices are of the form
\begin{equation}
\varrho=\left(\begin{array}{cccc}
a & d & e & 0 \\
d^* & b & f & 0 \\
e^* & f^* & c & 0 \\
0 & 0 & 0 & 0 \\ 
\end{array}\right)
\end{equation}
and only one element determines the tangle of such state, namely
\be
\tau(\varrho) = 4ff^* \, .
\ee
It follows that only the matrix element standing with the $\ket{01}\bra{10}$ 
term will be important for us. Direct calculations lead us to 
value $f=\alpha_j\alpha_k^*$ for the pair of $j$th and $k$th, hence
\be
\tau_{jk}=4|\alpha_j|^2 |\alpha_k|^2 \, .
\ee
In the next step we evaluate the tangle between the $j$-th
particle and the rest of the system. The state of single particle
is described by matrix
\begin{equation}
\varrho_j = \left(\begin{array}{cc}
|\alpha_0|^2 +\sum_{k\ne j} |\alpha_k|^2 & \alpha_0\alpha_j^* \\
\alpha_j\alpha_0^* & |\alpha_j|^2 
\end{array}\right)
\end{equation}
Now it is easy to check that
\be
\tau_{j} = 4\det\varrho_j = |\alpha_j|^2\sum_{k\ne j}|\alpha_k|^2 =\sum_{k\ne j} \tau_{jk}
\ee
and therefore the CKW inequalities are saturated, i.e. 
$\Delta_j =\tau_{j} -\sum_{k\ne j} \tau_{jk}=0$ for all $j$.
\end{example}

\subsection{Entanglement in partial swap collision model}
A general pure input state equals
\be
\ket{\Psi}=\ket{\psi}_S\otimes\ket{\varphi}_1\otimes\cdots\otimes\ket{\varphi}_N\,.
\ee
Since $[I,V\otimes V]=[S,V\otimes V]=0$ for all unitary operators $V$
on $\cH$ it follows that also $[U_\eta,V\otimes V]=0$ for all unitary $V$.
Moreover, the value of entanglement is invariant under local unitary 
transformations. Therefore, without loss of generality 
we can set $\ket{\varphi}=\ket{0}$ and 
$\ket{\psi}=\alpha\ket{0}+\beta\ket{1}$. After $n$ interactions
\be
\nonumber
\ket{\Psi_n}=\alpha\ket{0^{\otimes(N+1)}}
+\beta c^n\ket{1}_S\otimes \ket{0^{\otimes N}}+\beta
\sum_{l=1}^n
\ket{1}_l\otimes\ket{0^{\otimes N_{\overline{l}}}}
\left[isc^{l-1}e^{i\eta(n-l)}\right]\,,
\ee
where $N_{\overline{l}}$ denotes a system of $N$ qubits obtained by replacing
the $l$th qubit from the reservoir by the system qubit.

Since $\ket{\Psi_n}$ belongs to the class of states discussed
in Example \ref{ex:ckw}, we know that the state saturates CKW 
inequalities. We can used the derived formulas to write 
\be
\tau_{jk}(n)&=&4 |\beta|^4 s^4c^{2(j+k-2)}\,; \\
\tau_{0k}(n)&=&4 |\beta|^4 s^2c^{2(n+k-1)}\,;\\ 
\tau_{j}(n)&=&4|\beta|^4 s^2c^{2(j-1)}(1-s^2 c^{2(j-1)}) \,;\\
\tau_{0}(n)&=&4|\beta|^4 c^{2n}(1-c^{2n})\,.
\ee
Let us note that these formulas are valid only if $j,k\leq n$. Otherwise
the quantities vanish.

\begin{figure}
\begin{center}
\includegraphics[width=6cm]{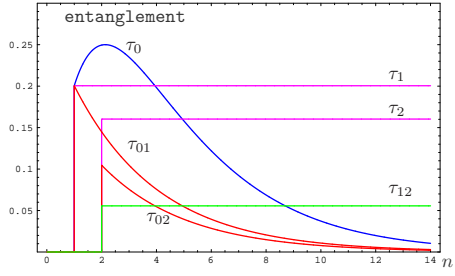}
\caption{Entanglement in the homogenization process.}
\end{center}
\end{figure}

These results show that the system particle acts as a mediator 
entangling the reservoir particles which have never interacted 
directly. It is obvious that the later the two reservoir qubits
interact with the system particle, the smaller the degree of their 
mutual entanglement is. Nevertheless, this value remains constant and 
does not depend on the subsequent evolution of the collision model 
(i.e., it does not depend on the number of interactions $n$). On the 
other hand the entanglement between the system particle and 
$k$th reservoir particle ($k$ arbitrary) increases at the moment
of interaction and then monotonously decreases 
with the number of interaction steps. 

\subsection{Entanglement in controlled unitary collision model}
As before, let us assume that initially the system and the reservoir
of $N$ qubits are described by a pure input state
\be
\ket{\Omega}=\ket{\psi}_S\otimes\ket{\varphi}_1\otimes\cdots\otimes\ket{\varphi}_N\,,
\ee
with $\psi=\alpha\ket{0}+\beta\ket{1}$.
After $n$ controlled unitary collisions we get
\be\ket{\Omega_n}=
  [\alpha\ket{0}\otimes\ket{\varphi_0^{\otimes n}}
    +\beta\ket{1}\otimes\ket{\varphi_1^{\otimes n}}
  ]\otimes\ket{\varphi^{\otimes (N-n)}}\,,
\ee
where we set $\ket{\varphi_0}=V_0\ket{\varphi}$ and
$\ket{\varphi_1}=V_1\ket{\varphi}$.

Tracing out the system we end up with the reservoir described by
the state
\be
\omega_{\rm env}(n)=
(|\alpha|^2\ket{\varphi_0^{\otimes n}}\bra{\varphi_0^{\otimes n}}+
|\beta|^2\ket{\varphi_1^{\otimes n}}\bra{\varphi_1^{\otimes n}})\otimes
\ket{\varphi}\bra{\varphi}^{\otimes (N-n)}
\,,
\ee
which is clearly separable, i.e. $\tau_{jk}(n)=0$. In this case no entanglement
is created within the environment.

A state of the system qubit and $k$th qubit from the reservoir equals
(providing that $n\geq k$)
\be
\nonumber
\varrho_{0k}(n)&=&|\alpha|^2\ket{0\varphi_0}\bra{0\varphi_0}+|\beta|^2\ket{1\varphi_1}\bra{1\varphi_1}\\
& & +\alpha\beta^*|\ip{\varphi_0}{\varphi_1}|^{(n-1)}\ket{0\varphi_0}\bra{1\varphi_1}+c.c\,,
\ee
thus for the tangle we have 
\be
\tau_{0k}(n)=4|\alpha|^2|\beta|^2|\ip{\varphi_0}{\varphi_1}|^{2(n-1)}
|\ip{\varphi_0}{\varphi_1^\perp}|^{2}\,.
\ee
If $n<k$ then $\tau_{0k}(n)=0$. Further,
\be
\varrho_0(n)&=&|\alpha|^2\ket{0}\bra{0}+|\beta|^2\ket{1}\bra{1}+
\alpha\beta^*|\ip{\varphi_0}{\varphi_1}|^{n}\ket{0}\bra{1}+c.c.\,,\\
\varrho_k(n)&=&|\alpha|^2\ket{\varphi_0}\bra{\varphi_0}+|\beta|^2\ket{\varphi_1}\bra{\varphi_1}\,,\\
\ee
resulting in the following values for tangle
\be
\tau_0(n)&=&4|\alpha|^2|\beta|^2(1-|\ip{\varphi_0}{\varphi_1}|^{2n})\,,\\
\tau_k(n)&=&4|\alpha|^2|\beta|^2|\ip{\varphi_0}{\varphi_1^\perp}|^{2}\,.
\ee

\begin{figure}
\begin{center}
\includegraphics[width=6cm]{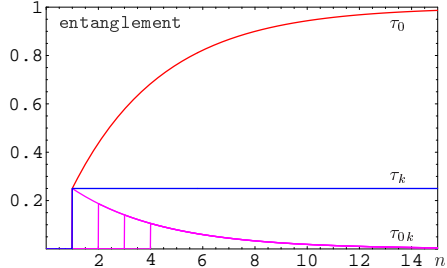}
\caption{Entanglement in the decoherence process.}
\end{center}
\end{figure}

In summary, during the decoherence collision model the qubits in the 
reservoir are not entangled at all. Interaction of the system qubit with
an individual reservoir's qubit entangles these pair of qubits. However,
with the number of collision this bipartite entanglement vanishes 
and finally $\tau_{0k}=0$ for all reservoir's qubits $k$. The entanglement
between a fixed qubit of the reservoir and rest of the qubits preserves its
value it gathered after the interaction. The entanglement between a system qubit
and the reservoir is increasing with the number of collisions and approaches 
the maximal value $\tau_0(\infty)=4|\alpha|^2|\beta|^2$.

\section{Master equations for collision models}\label{aba:sec5}
Within the collision model the system undergoes a discrete dynamics
\be
\varrho\to\cE_1[\varrho]\to\cE_2[\varrho]\to\cdots\to
\cE_n[\varrho]\to\cdots\,,
\ee
where $\cE_k=\cE\circ\cdots\circ\cE=\cE^k$. Let us define 
$\cE_0$ as the identity map $\cI$. It is straightforward to
see that the sequence of channels $\cE_0,\cE_1,\cE_2,\dots$ form a discrete
semigroup satisfying the relation
\be
\cE_{n}\circ\cE_m=\cE_{n+m}
\ee
for all $n,m=0,1,2,\dots$. In this section we will
investigate the continuous interpolations of such 
discrete semigroups. Interpreting the continuous parameter as time
we will derive first order differential master equation generating the
continuous sequence of linear maps. Our derivation of the master equation 
will be purely heuristic. We will simply replace $n$ in $\cE_n$ by 
a continuous time parameter $t$ and after that we will have a look what 
are the consequences. 

Let $\varrho_t=\cE_t[\varrho_0]$. The time derivative of the state dynamics 
results in the following differential equation
\be
\label{eq:derivation}
\frac{d\varrho_t}{dt}=\left(\frac{d\cE_t}{dt}\right)[\varrho_0]=
\left(\frac{d\cE_t}{dt}\circ\cE_t^{-1}\right)[\varrho_t]=\cG_t[\varrho_t]\,,
\ee
where $\cG_t$ is the generator of the dynamics $\cE_t$. Let us note that
this approach is very formal. There is no guarantee that the mapping
$\cE_t$ is a valid quantum channel for each value of $t$ and that
the set $\cE_t$ is indeed continuous. Another problematic issue is 
the assumed existence of the inverse of the map $\cE_t^{-1}$ for 
all $t\geq 0$. All these properties need to be checked before we give
some credit to the derived master equation.

\subsection{One-parametric semigroups}
In a special case when $\cG_t$ is time-independent the solution
of the differential equation
\be
\frac{d\varrho_t}{dt}=\cG[\varrho_t]\,,
\ee
with the initial condition $\varrho_{t=0}=\varrho_0$ can be written in the form
\be
\varrho_t=e^{\cG t}[\varrho_0]\,;\quad\cE_t=e^{\cG t}\,.
\ee
Since $e^{\cG (t+s)}=e^{\cG s}e^{\cG t}$ it follows that 
\be
\cE_t\circ\cE_s=\cE_{t+s}
\ee
for all $t,s\geq 0$ and $\cE_0=\cI$. 
In such case the continuous set of linear maps
$\{\cE_t\}$ form \emph{a one-parametric semigroup}, which are playing 
important role in the theory of open system dynamics. Semigroups of channels
are often used to approximate real open system's evolutions. In particular,
they are satisfying the so-called \emph{Markovianity condition}\cite{spohn} 
restricting the memory effects of the environment. The following theorem
specifies the properties of the generator $\cG$ assuring the validity 
of the quantum channel constraints on $\cE_t=e^{t\cG}$.

\begin{theorem}({\it Lindblad}\cite{lindblad})
$\cE_t=e^{\cG t}$ are valid quantum channels for all $t\geq 0$ if and only
if $\cG$ can be written in the form
\be
\cG[\varrho]=-\frac{i}{\hbar}[H,\varrho]+\frac{1}{2}\sum_{jk} c_{jk}([\Lambda_j\varrho, \Lambda_k]+[\Lambda_j,\varrho \Lambda_k])\,,
\ee
where $\{I,\Lambda_1,\dots,\Lambda_{d^2-1}\}$ is an orthonormal 
operator basis consisting of hermitian traceless operators $\Lambda_j$, 
i.e. $\tr{\Lambda_j\Lambda_k}=\delta_{jk}$,
the coefficients $c_{jk}$ form a positive matrix and $H$ is a hermitian 
traceless operator on $\cH$.
\end{theorem}

\subsection{Divisibility of channels}
Any collision $U$ induced some channel $\cE$. Our interpolation
should smoothly connect this channel with the identity map $\cI$
such that the connecting line is inside the set of channels. Since
the set of channels is convex any two points are connected by a line
containing only proper quantum channels, i.e. interpolation is surely
possible although its continuity is still not guaranteed. Let us have a look
on the possibility to interpolate $\cI$ and $\cE$ such that $\cE=e^{\tau\cG}$
for some time $\tau$ and time-independent Lindblad generator. 
Without loss of generality we can set $\tau=1$, i.e. 
$\cE=e^\cG$. The semigroup property implies that $\cE=\cE_\varepsilon\circ
\cE_{1-\varepsilon}$ for any $0<\epsilon<1$, or even stronger
$\cE=\cE_{\epsilon}^{1/\epsilon}$.

It is a surprising fact\cite{wolf_divisibility} that there are channels
$\cE$, which cannot be expressed as a nontrivial composition of some other
channels $\cE_1,\cE_2$. Consequently not all channels can be part of some
one-parametric semigroup. Before we will show one example, let us formally
define the indivisibility.

\begin{definition}
We say a channel $\cE$ is \emph{indivisible} if $\cE=\cE_1\circ\cE_2$ implies
that either $\cE_1$, or $\cE_2$ are unitary channels (such decomposition
is called trivial).
\end{definition}

Let us note that one of seemingly counter-intuitive consequences
is that unitary channels are indivisible. In a suitable representation
the composition of channels is just a product of matrices. Thus for 
determinants
\be
\det(\cE_1\circ\cE_2)=(\det\cE_1)(\det\cE_2)\,.
\ee
Let us denote by $\cE_{\min}$ a channel minimizing the value of
determinant, i.e. $\cE_{\min}=\arg\min_\cE\det\cE$. Assume 
$\cE_{\min}=\cE_1\circ\cE_2$. Since the value of $\det\cE_{\min}$
is negative and for channels $|\det\cE|\leq 1$,
it follows that either $\det\cE_1=1,\det\cE_2=\det\cE_{\min}$, 
or $\det\cE_2=1,\det\cE_1=\det\cE_{\min}$. But $\det\cE=1$ implies 
$\cE$ is a unitary channel, which proves the triviality of any 
decomposition of the minimum 
determinant channel. If $d=2$ (qubit), then one of the minimum 
determinant channels is the 
\emph{optimal universal NOT}\cite{wolf_divisibility}, i.e.
\be
\cE_{\min}[\varrho]=\frac{1}{3}(I+\varrho^T)\,,
\ee
for which $\det\cE_{\min}=-1/27$. The mapping $\varrho\to\varrho^T$
known as the \emph{universal NOT} is not a proper quantum channel, because
it is not completely positive and, thus, unphysical. Let us note that channels
of minimal determinant are not the only indivisible channels. As far as the
authors know a complete characterization of indivisible channels
is not known.

For our purposes the introduced concept of indivisibility is not sufficient
as we are interested in the existence of continuous semigroup between $\cI$
and $\cE$ generated by a unitary collision $U$. We need $\cE$ to be
\emph{infinitely divisible}, i.e. $\cE=\cE_{1/n}^{n}$ for
all $n>0$. It was shown by Denisov\cite{denisov} that infinitely
divisible channel can be expressed in the form $\cE=\cF\circ e^\cG$,
where $\cF$ is an idempotent channel ($\cF^2=\cF$) satisfying 
$\cF\cG=\cF\cG\cF$ and $\cG$ is a Lindblad generator. We are interested 
in cases when $\cF=\cI$, thus, $\cE=e^\cG$. Then it follows that
the sequence of channels $\cE^n=e^{n\cG}$ determined by the collision model
belongs to the one-parametric semigroup $e^{t\cG}$. In other words a discrete 
semigroup is interpolated by a continuous one. We will see that this is the 
case for the collisions considered in the homogenization and decoherence
processes.

For a given channel induced by a single collision
a question of interest is whether $\cE=e^\cG$ for
some Lindblad generator $\cG$. In other words whether
$\cG=\log\cE$ is a valid generator of complete positive
dynamics. Since channels $\cE$ are not necessarily 
diagonalizable, the notion of the matrix/operator logarithm is
not completely trivial. The logarithm is not unique and any operator 
$X$ such that $e^X=\cE$ is called $\log\cE$. We need to search all
logarithms in order to find a valid generator of the Lindblad form.
Evaluating the logarithm will give us a time-independent generator 
(if it exists) interpolating the discrete semigroup $\cE_n=\cE^n$. 
For both considered collision models the evaluation of the 
logarithm is not difficult. However, for illustrative purposes
in our analysis of homogenization we will follow the {\it ad hoc} 
procedure exploiting 
\eqref{eq:derivation}, which (if successful) could capture also 
time-dependent generators. For the case of decoherence we will 
evaluate the logarithm of the induced channel.

\subsection{Bloch sphere parametrization}
The set of operators form a complex linear space endowed with a 
\emph{Hilbert-Schmidt scalar product} $(X,Y)=\tr{X^\dagger Y}$. 
Let $\varphi_1,\dots,\varphi_d$ be an orthonormal basis of $\cH$.
Then the set of $d^2$ operators $\{E_{jk}=\ket{\varphi_j}\bra{\varphi_k}\}$ 
form an orthonormal operator basis, i.e. $\tr{E_{j^\prime k^\prime}^\dagger E_{jk}}=
\ip{\varphi_{j^\prime}}{\varphi_{j}}\ip{\varphi_k}{\varphi_{k^\prime}}
=\delta_{jj^\prime}\delta_{kk^\prime}$. 

In what follows we will restrict to the case of two-dimensional
Hilbert spaces $\cH$ (qubits). This will be perfectly sufficient
for our later purposes, because we will focus on qubit collision models
of homogenization and decoherence. However, many of the
properties and procedures can be easily extended to more dimensional case.

Let us start with a very convenient operator basis of qubit systems - the 
set of \emph{Pauli operators} 
\be
\nonumber
&I=\ket{0}\bra{0}+\ket{1}\bra{1}\,,\qquad
\sigma_x=\ket{0}\bra{1}+\ket{1}\bra{0}\,,&\\
\nonumber
&\sigma_y=-i(\ket{0}\bra{1}-\ket{1}\bra{0})\,,\quad
\sigma_z=\ket{0}\bra{0}-\ket{1}\bra{1}\,,&
\ee
where $\{\ket{0},\ket{1}\}$ is an orthonormal basis of $\cH$. 
These operators are hermitian, $\sigma_j=\sigma_j^\dagger$, and mutually 
orthogonal, i.e. $\tr{\sigma_j\sigma_k}=2\delta_{jk}$ (we set 
$\sigma_0=I$). They are also unitary, but this property will not
be very important for us.

Using Pauli operators the density operators takes the form
\be
\varrho=\frac{1}{2}(I+\vec{r}\cdot\vec{\sigma})\,,
\ee
where $\vec{\sigma}=(\sigma_x,\sigma_y,\sigma_z)$ and
$\vec{r}=\tr{\varrho\vec{\sigma}}$ is the so-called \emph{Bloch vector}.
The positivity constraint restricts its length $|\vec{r}|\leq 1$, i.e.
in the Bloch vector parametrization the qubit's states
form a unit sphere called \emph{Bloch sphere}. Let us note that such 
simple picture of state space does not hold for more dimensional quantum 
systems.

A quantum channel $\cE$ is acting on Pauli operators as follows 
\be
\varrho\to\varrho^\prime=\cE[\varrho]=\frac{1}{2}(\cE[I]+
x\cE[\sigma_x]+y\cE[\sigma_y]+z\cE[\sigma_z])\,.
\ee
Using $\cE[\sigma_j]=\sum_k\cE_{kj}\sigma_k$ with $\cE_{kj}=\frac{1}{2}
\tr{\sigma_k\cE[\sigma_j]}$. The tracepreservity ensures that
$\cE_{00}=\frac{1}{2}\tr{\cE[I]}=1$ 
and $\cE_{0j}=\tr{\cE[\sigma_j]}=0$ for $j=x,y,z$.
That is,
\be
\varrho^\prime=\frac{1}{2}\left[I+
\sum_{k=x,y,x}\left(
\cE_{k0}+\sum_{j=x,y,x} \cE_{kj}r_j\right)\sigma_k\right]\,.
\ee
Comparing with the expression 
$\varrho^\prime=\frac{1}{2}(I+\vec{r}^\prime\cdot\vec{\sigma})$
we get that under the action of a channel $\cE$ the Bloch vectors
are transformed by an affine transformation
\be
\vec{r}\to\vec{r}^\prime=\vec{t}+T\vec{r}\,,
\ee
where $t_j=\cE_{j0}$ and $T$ is a 3x3 matrix with entries
$T_{jk}=\cE_{jk}$ for $j,k=x,y,z$. 

The Lindblad generator takes the form of $4\times 4$ matrix
\be
\cG=\left(
\begin{array}{cccc}
0 & 0 & 0 & 0\\
g_{10} & g_{11} & g_{12} & g_{13}\\
g_{20} & g_{21} & g_{22} & g_{23}\\
g_{30} & g_{31} & g_{32} & g_{33}
\end{array}
\right)\,,
\ee
where $g_{jk}=\frac{1}{2}\tr{\sigma_j\cG[\sigma_k]}$.
Inserting the Lindblad operator-sum form 
\be
\cG[X]=-\frac{i}{\hbar}\sum_{j=x,y,z} h_j[\sigma_j,X]+\frac{1}{2}\sum_{j,k=x,y,z} 
c_{jk}([\sigma_j,X\sigma_k]+[\sigma_j X,\sigma_k])\,.
\ee
into the matrix expression we obtain
\begin{eqnarray}
g_{jk} &=& 2\sum_l\epsilon_{jkl}h_l
+\frac{1}{2}(c_{kj}+c_{jk})-\sum_l c_{ll}\delta_{jk}\,, \\
g_{k0} &=& i\sum_{jl}\epsilon_{jlk} c_{jl}\,.
\end{eqnarray}
The inverse relations express the parameters $c_{jk}$ and $h_j$
via the elements of the matrix $\cG$
\be
\nonumber
& 
h_1 =  \frac{g_{32} -g_{23}}{4}\,,\quad
h_2  =  \frac{g_{13} -g_{31}}{4}\,,\quad
h_3  =  \frac{g_{21} -g_{12}}{4} & \\
\label{eq:prechod}
&c_{jj}=g_{jj}-\frac{1}{2}\sum_{k} g_{kk}&\\\nonumber
& c_{12}=\frac{1}{2}(g_{12}+g_{21}-ig_{30})\,,\quad
c_{21}=\frac{1}{2}(g_{12}+g_{21}+ig_{30})\,,&\\\nonumber
& c_{23}=\frac{1}{2}(g_{23}+g_{32}-ig_{10})\,,\quad
c_{32}=\frac{1}{2}(g_{23}+g_{32}+ig_{10})\,,&\\\nonumber
& c_{13}=\frac{1}{2}(g_{13}+g_{31}+ig_{20})\,,\quad
c_{31}=\frac{1}{2}(g_{13}+g_{31}-ig_{20})\,.&
\ee

\subsection{Master equation for homogenization collision model}
The collision model of quantum homogenization driven by partial swaps
results in the sequence of channels being powers of a channel $\cE$
defined in Eq.\eqref{2.4}, i.e.
\be
\label{eq:bloch_homo}
\vec{r}\to\vec{r}^\prime=c^2\vec{r}+s^2\vec{w}-cs\vec{w}\times\vec{r}\,,
\ee
where $\vec{w}$ is the Bloch vector associated with the state $\xi$ and
$c=\cos\eta,s=\sin\eta$. Choosing a new operator basis 
$S_j=V\sigma_j V^\dagger$ such that $\xi=\frac{1}{2}(I+wS_3)$ 
we get
\be
\cE=\left(
\begin{array}{cccc}
1 & 0 & 0 & 0\\
0 & c^2 & csw & 0\\
0 & -csw & c^2 & 0\\
s^2 w & 0 & 0 & c^2
\end{array}
\right)\,.
\ee
Let us note that the basis transformation $\sigma_j\to S_j$ corresponds 
to a rotation of the coordinate system of the Bloch sphere 
representation. Moreover, we used that the partial swap commutes with
unitaries of the form $V\otimes V$, hence the form of the induced map 
\eqref{eq:bloch_homo} is unaffected, only the vectors are expressed 
with respect to different coordinate system. 

Let us introduce an angle $\theta=\arctan(ws/c)=\arctan(w\tan\eta)$
and a parameter $q=\sqrt{c^2+w^2s^2}$. Using the identity 
$\cos\arctan(x)=1/\sqrt{1+x^2}$ we get $q\cos\theta=\cos\eta=c$
and $q\sin\theta=w\sin\eta$, thus
\be
\cE=\left(
\begin{array}{cccc}
1 & 0 & 0 & 0\\
0 & cq\cos\theta & cq\sin\theta & 0\\
0 & -cq\sin\theta & cq\cos\theta & 0\\
s^2 w & 0 & 0 & c^2
\end{array}
\right)\,.
\ee
The reason for such parametrization becomes clear if we evaluate
powers of $\cE$
  \begin{equation}
    \nonumber  \cE^n=\left(
  \begin{array}{cccc}
    1 &0 & 0 & 0 \\
    0 & c^{n}q^n\cos{n\theta} & c^nq^n\sin{n\theta} & 0\\
    0 & -c^nq^n\sin{n\theta} & c^{n}q^n\cos{n\theta} & 0\\
    w (1-c^{2n}) & 0 & 0 & c^{2n}
    \end{array}
  \right)\,.
  \end{equation}

In the next step we make \emph{ad hoc} assumption and replace 
$n$ by $t/\tau$, where $\tau$ is some time scale and $t\geq 0$ is 
a continuous time parameter. Introducing the parameters
\be
\Omega=\theta/\tau\,,\quad 
c^{2t/\tau}=e^{-\Gamma_1 t}\,,\quad
(cq)^{t/\tau}=e^{-\Gamma_2 t}\,,
\ee
we end up with the continuous set of trace-preserving linear maps
  \begin{equation}
  \cE_t=\left(
  \begin{array}{cccc}
    1 &0 & 0 & 0 \\
    0 & e^{-\Gamma_2 t}\cos{\Omega t} & e^{-\Gamma_2 t}\sin{\Omega t} & 0\\
    0 & -e^{-\Gamma_2 t}\sin{\Omega t} & e^{-\Gamma_2 t}\cos{\Omega t} & 0\\
    w (1-e^{-\Gamma_1 t}) & 0 & 0 & e^{-\Gamma_1 t}
    \end{array}
  \right)\,.
  \end{equation}
It is not difficult to see that they form a one-parametric 
semigroup ($\cE_t\circ\cE_s=\cE_{t+s}$), but one needs to verify 
whether for each $t$ these maps are completely positive. Therefore,
we calculate the generator $\cG$ and verify its properties. In particular,
we obtain a time-independent generator
\be
    \nonumber  \cG=\left(
  \begin{array}{cccc}
    0 &0 & 0 & 0 \\
    0 & -\Gamma_2 & -\Omega & 0\\
    0 & \Omega & -\Gamma_2 & 0\\
    w\Gamma_1 & 0 & 0 & -\Gamma_1
    \end{array}
  \right)\,.
\ee
Using the Eqs.\eqref{eq:prechod} we can verify that the parameters
$h_1=h_2=0,h_3=-\Omega/2$ are real and that the matrix
\be
C=\left(\begin{array}{ccc}
-\Gamma_1/2 & -i2w\Gamma_1 & 0 \\
i2w\Gamma_1 & -\Gamma_1/2 & 0 \\
0 & 0 & -(\Gamma_1+2\Gamma_2)/2 \\
\end{array}\right)
\ee
is positive, hence, the generator $\cG$ is of Lindblad form.
In the operator-sum (Kraus) form the master equation reads
\be
\frac{d\varrho}{dt}&=&i\frac{\Omega}{2\hbar}[S_z,\varrho]
-iw\Gamma_1 (\Phi_{xy}[\varrho]+\Phi_{yx}[\varrho])\\\nonumber
& & -\frac{\Gamma_1}{4}(\Phi_{xx}[\varrho]+\Phi_{yy}[\varrho])
-\frac{\Gamma_1+2\Gamma_2}{4}\Phi_{xx}[\varrho]\,,
\ee
where $\Phi_{jk}[\varrho]=\frac{1}{2}([\sigma_j\varrho,\sigma_k]+[\sigma_j,\varrho \sigma_k])$. After a little algebra we obtain
\be
\frac{d\varrho}{dt}&=&i\frac{\Omega}{2\hbar}[S_z,\varrho]
  -iw\Gamma_1(S_x\varrho S_y-S_y\varrho S_x+i\varrho S_z+iS_z\varrho)\\
  \nonumber & &
  +\frac{1}{4}\Gamma_1(S_x\varrho S_x+S_y\varrho S_y-2\varrho)
  +\frac{1}{4}(2\Gamma_2-\Gamma_1)(S_z\varrho S_z-\varrho)\,.
  \ee

Let us note that for the special choice of parameters the above master
equation coincide with the master equation used to model the
spontaneous decay of a two-level atom. In particular,
setting $\Gamma_1=2\Gamma_2=2\gamma=-\frac{2}{\tau}\ln\cos\eta$ 
and $\xi$ being a pure state ($w=1$) we get
\be
\frac{d}{dt}\varrho=-i\frac{\Omega}{2\hbar}[S_z,\varrho]+\frac{\gamma}{2}[2S_-\varrho S_+-S_-S_+\varrho-\varrho S_+S_-)
\ee
where we used $S_{\pm}=(S_x\pm iS_y)/2$.

\subsection{Master equation for decoherence collision model}
In this section we will repeat the same steps as in the previous one,
but for collisions described by controlled unitary interactions. Without loss
of generality we will assume that the decoherence basis coincide with 
the eigenvectors of $S_z=V\sigma_z V^\dagger$ operator. The channel
induced by a collision $U=\ket{0}\bra{0}\otimes V_0+\ket{1}\bra{1}\otimes V_1$
reads
\be
\cE=
\left(
\begin{array}{cccc}
1 & 0 & 0 & 0 \\
0 & \lambda\cos\varphi &\lambda\sin\varphi & 0\\
0 & -\lambda\sin\varphi &\lambda\cos\varphi & 0\\
0 & 0 & 0 & 1
\end{array}
\right)\,,
\ee
where $\tr{V_1^\dagger V_0\xi}=\lambda e^{i\varphi}$. We can follow the same
derivation of master equation as in the case of homogenization 
\cite{ziman_decoherence}. However, for the illustration purposes
we will describe the second procedure based on evaluation of $\log\cE$.
Unlike the case of homogenization the decoherence channel $\cE$
defines a hermitian matrix, thus, the logarithm is pretty easy to 
calculate exploiting the simple functional calculus for hermitian 
operators. 

Let us start with the observation that
\be
\lambda(\cos\varphi I+i\sin\varphi \sigma_y)=e^{\ln\lambda I} 
e^{i(\varphi+2k\pi)\sigma_y}=e^{\ln\lambda I+i(\varphi+2k\pi)\sigma_y}\,,
\ee
where $k$ is arbitrary natural number. For $k=0$ the logarithm is called
principal. For our purposes it is sufficient to consider only this one,
because we are restricted to angles inside the interval $[0,2\pi]$.
Since this matrix is the central part of the matrix $\cE$, it follows that
\be
\cG=\log\cE=\left(
\begin{array}{cccc}
0 & 0 & 0 & 0 \\
0 & \ln\lambda & -\varphi & 0\\
0 & \varphi &\ln\lambda & 0\\
0 & 0 & 0 & 0
\end{array}
\right)\,,
\ee
is the generator of the semigroup dynamics containing the channel 
$\cE$. Is $\cE_t=e^{\cG t}$ a valid quantum channel for any $t$?

Using the relations specified in Eqs.\eqref{eq:prechod} we found that
the only nonzero parameters are 
\be
h_3=\varphi/2\,,\qquad c_{33}=-\ln\lambda/2\,.
\ee
Since $c_{33}\geq 0$ it follows that the matrix $C$ is positive.
Therefore
\be
\nonumber
\cG[\varrho]&=&-i\frac{\varphi}{2\hbar}[S_z,\varrho]-\frac{\ln\lambda}{4}
([S_z,\varrho S_z]+[S_z\varrho,S_z])\\
\nonumber
&=&-i\frac{\varphi}{2\hbar}[S_z,\varrho]-\frac{\ln\lambda}{2}
(S_z\varrho S_z-\varrho)\,,
\ee
defines a correct Lindblad generator. Using 
$H=\frac{1}{2}\varphi S_z$ and $\gamma=(2\ln\lambda)/\varphi^2$
we obtain a well-known master equation
\be
\frac{d\varrho}{dt}
=-\frac{i}{\hbar}[H,\varrho]-\frac{\gamma}{2}
[H,[H,\varrho]]\,,
\ee  
which is used to model the decoherence also for more dimensional systems.

\section{Conclusions}\label{aba:sec6}
In these lectures we introduced and investigated the simple collision 
model to capture relevant features and properties of quantum open system 
dynamics. We focused on two particular collision models determined by
the choice of the unitary transformations describing the individual
collisions:
\begin{itemize}
\item{\it Homogenization} induced by the partial swap interactions 
$U_\eta=\cos\eta I+i\sin\eta S$. This process was motivated by the classical
thermalization process (0th law of thermodynamics).

\item{\it Decoherence} induced by the controlled unitary interactions
$U=\sum_j \ket{\varphi_j}\bra{\varphi_j}\otimes V_j$. This process
describes the disappearance of the quantumness of quantum systems.
\end{itemize}

It is known that creation of entanglement requires interactions between
quantum systems. In our collision model we start from completely factorized 
state. It is an interesting question what type of multipartite entanglement
is created via a well-defined sequence of bipartite collisions. We haven't
provided a definite answer to this problem, but we illustrated that the
created entanglement of the particular collision models is not trivial.
In particular, the partial swap interaction creates $W$-type of entanglement
for which all the involved systems are pairwisely entangled although they
did not interact directly. Moreover, during the whole evolution the
CKW inequalities are saturated. The entanglement created in the decoherence
collision model is of completely different quality. In this case the created
entanglement is of $GHZ$-type meaning that pairwise entanglement in the 
reservoir vanishes, however the system is still entangled. The reservoir
itself (the system is traced out) is in a separable state. The collision models
can be understood as a (simple) preparation processes aiming to create 
multipartite entanglement. And it is of interest to understand what quality 
and quantity of entanglement it is capable of.

The well-defined sequence of collisions and specific initial conditions
imply that the discrete time evolution of the system is described by
a discrete semigroup of natural powers of $\cE$. May this
this discrete  set of channels be interpolated by a single 
one-parametric continuum of channels? Is this continuum 
a semigroup of channels, or not? 
Luckily, we have derived that the interpolating continuum for 
both cases of homogenization and decoherences. By deriving the
corresponding master equations and testing the Markovianity of the generator
we showed that these sets form one-parametric semigroups 
of channels. Let us note that a collision model simulation of 
given semigroups of channels is easy. Each semigroup $\cE_t$ can
be discretized by introducing a parameter $\tau$ and set 
$\cE_n=\cE_{n\tau}$. Since the choice of an interaction $U$
inducing $\cE_\tau$ is not unique, also the collision model 
is not unique. However, the particular collision dynamics
of the system will be the same. 

Let us note an interesting fact. Fix a unitary collision $U$. Even if
$\xi,\xi^\prime$ induce a channels $\cE,\cE^\prime$ with valid generators
$\cG,\cG^\prime$, the convex combination $\lambda\xi+(1-\lambda)\xi^\prime$ 
does not have to be associated with a correct Lindblad generator. There 
are collision models that cannot
be interpolated by Markovian dynamics. Surprisingly, collision models
provides richer dynamics than Lindblad's master equations. To sort
out example consider 
Markovian models developed in order to describe true open system dynamics 
can be efficiently simulated (approximated) in collision models.
We believe that such toy collision models provide an interesting playground
capturing all the conceptual features of quantum open system dynamics.

\section*{Acknowledgments}
We would like to thank coauthors of the original papers 
Peter \v Stelmachovi\v c, Valerio Scarani, Nicolas Gisin, and Mark Hillery. 
Many thanks also to Daniel Burgarth for inspiring discussions on closely 
related topics, which would deserve some space in these lectures, 
but due to space-time constraints finally left outside the scope. 
We acknowledge financial support 
of the European Union projects HIP FP7-ICT-2007-C-221889, and
CE QUTE ITMS NFP 26240120009, and of the projects, CE SAS QUTE, 
and MSM0021622419.

\bibliographystyle{ws-procs9x6}
\bibliography{ws-pro-sample}

\begin{thebibliography}{9}
\bibitem{ziman_homogenization}
M.Ziman, P.\v Stelmachovi\v c, V.Bu\v zek, M.Hillery, V.Scarani, N.Gisin,
Dilluting quantum information: An analysis of information transfer in system-reservoir interactions,
Phys.Rev. A 65 , 042105 (2002).

\bibitem{scarani_thermo}
V.Scarani, M.Ziman, P.\v Stelmachovi\v c, N.Gisin, V.Bu\v zek,
Thermalizing Quantum Machines: Dissipation and Entanglement,
Phys.Rev.Lett. 88 , 97905-1 (2002).

\bibitem{ziman_open}
M.Ziman, P.\v Stelmachovi\v c, V.Bu\v zek,
Description of quantum dynamics of open systems based on collision-like models, 
Open systems and information dynamics 12, No.1, pp. 81-91 (2005).

\bibitem{ziman_decoherence}
M.Ziman, V.Bu\v zek,
All (qubit) decoherences: Complete characterization and physical implementations,
Phys.Rev.A 72, 022110 (2005).

\bibitem{schrodinger1926} 
E.Schr\"odinger,
An Undulatory Theory of the Mechanics of Atoms and Molecules,
Phys. Rev. 28, 1049–1070 (1926)

\bibitem{davies1970}
E.B.Davies,
{\em  Quantum Theory of Open Systems} (Academic, London, 1976).

\bibitem{alicki}
R.Alicki, K.Lendi,
  {\em Quantum Dynamical Semigroups and Applications},
  Lecture Notes in Physics (Springer-Verlag, Berlin, 1987).

\bibitem{breuer}
H.P.Breuer,
{\it The Theory Of Open Quantum Systems}
(Oxford University Press, USA, 2002).

\bibitem{nielsen}
M.A.Nielsen, I.L.Chuang,
{\em Quantum Computation and Quantum Information}
(Cambridge University Press, Cambridge, 2000).

\bibitem{heinosaari}
T.Heinosaari, M.Ziman,
Guide to mathematical concepts of quantum theory,
Acta Physica Slovaca 58, 487-674 (2008).

\bibitem{buzek_cloning}
V. Bu\v zek and M. Hillery,
Quantum copying: Beyond the no-cloning theorem,
Phys.Rev.A 54, 1844 (1996).

\bibitem{scarani_cloning}
V.Scarani, S.Iblisdir, N.Gisin, A.Acin,
Quantum cloning,
Rev.Mod.Phys. 77, 1225-1256 (2005) 

\bibitem{Simon&Reed80}
M. Reed, B. Simon, {\em \it Methods of Modern Mathematical
Physics I: Functional Analysis} (Academic Press, San Diego, 1980)

\bibitem{zurek}
W.H. Zurek,
Decoherence, einselection, and the quantum origins of the classical,
Rev.Mod.Phys. {\bf 75}, 715 (2003).

\bibitem{dariano_decoherence}
F.Buscemi, G.Chiribella, G.M.D'Ariano,
Inverting quantum decoherence by classical feedback from the environment,
Phys.Rev.Lett. 95, 090501 (2005).

\bibitem{horodecki_ent}
R.Horodecki, P.Horodecki, M.Horodecki, and K.Horodecki,
Quantum entanglement,
Rev.Mod.Phys. 81, pp. 865-942 (2009).

\bibitem{plenio_ent}
M.B.Plenio, S.Virmani,
An introduction to entanglement measures,
Quant.Inf.Comp. 7, 1 (2007).

\bibitem{wootters}
W.K.Wootters,
Entanglement of Formation of an Arbitrary State of Two Qubits,
Phys.Rev.Lett. 80, 2245 (1998).  

\bibitem{coffman}
V.Coffman, J.Kundu, W.K.Wootters, 
Distributed Entanglement,
Phys.Rev.A 61, 052306 (2000).

\bibitem{osborne}
T.J.Osborne, F.Verstraete,
General Monogamy Inequality for Bipartite Qubit Entanglement,
Phys.Rev.Lett. 96, 220503 (2006). 

\bibitem{spohn}
H. Spohn,
Kinetic equations from Hamiltonian dynamics: Markovian limit,
Rev.Mod.Phys. 53, 569 (1980).

\bibitem{lindblad}
G. Lindblad,
On the generators of quantum dynamical semigroups,
Comm.Math.Phys. 48, pp. 119-130 (1976). 

\bibitem{wolf_divisibility}
M.M.Wolf, J.I.Cirac:
Dividing quantum channels,
Comm.Math.Phys. 279, pp. 147-168 (2008)

\bibitem{denisov}
L.V.Denisov,
Infinitely divisible Markov mappings in the quantum theory of probability,
Th. Prob. Appl. 33, 392 (1988)

\end{thebibliography}

\end{document}